\documentclass[18pt]{amsart}
\usepackage{hyperref}
\usepackage[usenames]{color}
\usepackage{amsmath,amsthm,amsfonts,amssymb}

\newtheorem{theorem}{Theorem}
\newtheorem{lemma}[theorem]{Lemma}

\newtheorem{proposition}[theorem]{Proposition}

\theoremstyle{definition}

\newtheorem{example}[theorem]{Example}

\newtheorem{remark}[theorem]{Remark}

\numberwithin{equation}{section}

\begin{document}

\title{Multiparty secret sharing based on hidden multipliers}

\author{Vitaly Roman'kov}

\maketitle
\footnote{The study was carried out within the framework of the state contract
of the Sobolev Institute of Mathematics (project no. 0314-2019-0004).}
\begin{center}
Sobolev Institute of Mathematics of RAS (Omsk Branch), Omsk, Russia\\ romankov48@mail.ru
\end{center}

{\bf Abstract:} 

Secret sharing schemes based on the idea of hidden multipliers in encryption  are proposed. As a platform, one can use both multiplicative groups of finite fields and groups of invertible elements of commutative rings, in particular, multiplicative groups of residue rings. We propose two versions of the secret sharing scheme and  a version of ($k,n$)-thrested scheme. For a given $n$, the dealer can choose any $k.$ The main feature of the proposed schemes is that the shares of secrets are distributed once and can be used multiple times. This property distinguishes the proposed schemes from the secret sharing schemes known in the literature. The proposed schemes are semantically secure. The same message can be transmitted in different forms. From the transferred secret $c$ it is impossible to determine which of the two given secrets $m_1$ or $m_2$ was transferred.
For concreteness, we give some numerical examples. 


\section{Introduction}
\label{sec:1}
 Sharing secrets was introduced in 1979 by Shamir \cite{Shamir} and Blakely \cite{Bl}.
Since then, many applications have emerged for several different types of cryptographic protocols.
At the same time, research on some of the major  open problems in secret sharing has led to the development of a rich mathematical theory  related  to combinatorics, information theory,  algebra, and other fields. A secret sharing scheme involves a dealer who holds a secret. This dealer distributes pieces of its secret (called shares) to a set of participants.  Only
qualified participant groups that form the schema's access structure can recover
a complete secret using the shared resources available to them, and therefore only they can read the transmitted message.
  
This paper discusses the unconditionally secure committed
secret sharing schemes. These are the schemes in which the  unqualified sets of participants cannot  obtain any partial information about the transmitted message. The reader is referred to  \cite{Sti} for an introduction to
secret sharing. See also survey \cite{Bei}  and papers  \cite{CSB} -- \cite{HC}  for some secret sharing schemes.

The main goal of this article is to build schemes that provide for the reuse of once distributed secret resources. This is possible only in cases where the allocated secret shares are not revealed when the complete secret is built on their basis. It should also be possible to add, remove or replace qualified group members without changing their shares of the secret. These properties show the advantages of the proposed scheme in comparison with the known secret distribution schemes. The proposed schemas are semantically secure. The same message can be transmitted in different forms. From the transferred secret $c$ it is impossible to determine which of the two given secrets $m_1$ or $m_2$ was transferred.

For example, in Shamir's scheme the secret is a polynomial function $f(x).$ If $n$ is the degree of $f(x)$, then for its  restoration it is required to know $f(x_i)$'s at $n$ different points $x_1, \ldots , x_n.$ Therefore, any subset of at least $n$ members (provided that their at least $n$ points are different) will be qualified. Any subset of $m \leq n-1$ members is unqualified. However, once the secret is recovered, it becomes known to all participants and cannot be reused. 

We propose  schemes such that the initial distribution of shares of the secret between all participants in the process is carried out either using a secure communication channel, or using the protocol of secret key transfer over an open communication channel. The entire further process is carried out over an public network.

If $m \leq n$ then
an ($m, n$)-threshold scheme is one with $n$ total participants and in which any $m$
participants can combine their shares and recover the secret but not fewer than $m$.
The number $m$ is called the threshold. It is a secure secret sharing scheme if
given less than the threshold there is no chance to recover the secret. Using our approach, we also propose an ($m, n$)-scheme for arbitrary parameters $m, n; m\leq n.$ For a given $n$, the dealer can choose any $m$ without  redistributing  shares. This scheme is monotone, i.e., any set of $k\geq m$ participants is qualified.

\bigskip 
 Notation:  $\mathbb{Z}$ -- set of integer numbers, 
 $\mathbb{Z}_n = \mathbb{Z}/n\mathbb{Z}$ -- residue ring,  $\mathbb{N}$ -- set of nonnegative integer numbers, $\mathbb{N}_k = \{1, \ldots , k\}.$ 

\section{Construction of fields with prescribed orders of subgroups of multiplicative groups}
\label{sec:2}
The main idea behind the corresponding algorithm  is the following statement. 

\begin{proposition} (\cite{Hand}, Fact 4.59).
\label{pro:1}
  Let $n\geq 3$ be an odd integer, and suppose that 
$n=1+rq$, where $q$ is an odd prime and $r$ is even positive integer. 
Suppose further that $r < q$.
\begin{itemize}
\item
If there exists an integer $a$ satisfying $a^{n-1}\equiv 1 (\bmod\,n)$  and gcd($a^{r}-1, n$) = $1$, 
then $n$ is prime.
\item If $n$ is prime, the probability that a randomly selected base $a, 1 \leq a \leq n-1$, 
satisfies $a^{n-1}\equiv 1 (\bmod\,n)$  and gcd($a^{r}-1, n$) = $1$ is $\frac{q-1}{q}$.
\end{itemize}
\end{proposition}
Maurer's algorithm for generating provable primes (Algorithm 4.62 in \cite{Hand})  recursively generates an odd prime $n$, and then chooses random even integers $r, r < q$, until $n = 1 +rq$ 
 can be proven prime using  for some base $a$. By
proposition \ref{pro:1} the proportion of such bases is $1 - 1/q$  for prime $n$. On the other hand, if $n$ is
composite, then most bases $a$ will fail to satisfy the condition $a^{n-1}\equiv 1 (\bmod\,n)$.

Let's describe this algorithm.

\begin{enumerate}
\item We start with an  odd prime $q = q_1$.
\item Let's choose an even $r$ at random:
$$
q \leq r \leq 4q + 2.
$$
\item Consider
$$
n = qr + 1,\, 
q\leq r \leq 4q+2.
$$
\item
 Choose randomly the number $a = a_1$ within $1 <
 a < n-1$ and check the fulfillment of conditions  from
proposition \ref{pro:1}. If $a = a_1$ does not satisfy these conditions, then we take
another random number $a = a_2$. So we repeat a sufficient number of times:
$a = a_1, a_2, \ldots, a_k$ until we find a suitable value
$a$.

If you succeed in doing this, then $n$ is prime. We put $q = q_2 = n$ and
repeat the construction starting from first step. We do this until
don't get a big enough prime.
\end{enumerate}
If, with a large number of trials for $a$, it was not possible to execute the
conditions  of proposition \ref{pro:1}, then we change $r$ and repeat everything again.

Suppose that the constructed number $n$ is indeed
prime. Let us ask ourselves what is the probability of finding the number
$a$ with the given properties  from proposition \ref{pro:1}.  First,
note that the condition $a^{n-1} = 1 (\bmod\, n) $ will be automatically satisfied by Fermat's Little Theorem.
In this case
(for prime $n$) the second condition  gcd($a^{n-1/q} - 1, n) = 1$
$\sim$ gcd$(a^r-1 , n)=1$ is satisfied if and only if
 $a^r \neq 1 (\bmod\, n)$. In the field $ \mathbb{F}_n$, the equation $x^r = 1 $
has at most $r$ roots, one of which is equal to $1$, and the second
$n-1$. Therefore, on the interval $ 1 < a < n-1 $, there are at most $ r-2 $
numbers $ r $ for which $ a^r = 1$ in the field $ \mathbb{F}_n.$ This means
that the probability of choosing such $a$ is no more than $\frac{r-2}{n-3}
\sim \frac{r}{qr} = \frac {1}{q}$.
 
Note also that the thus constructed prime $n$ will be
is greater than $q^2$ because $q \leq r$ and $n = qr + 1$. When
sequential construction of primes $q_1, q_2, \ldots $ they
grow no less than quadratically.

Let's ask another question: how realistic is it to find
prime number $n = qr + 1 $ under the indicated constraints $ q \leq r \leq
4q + 2$, choosing an even $r$.

First of all, note that , by the famous Dirichlet theorem, the progression
$n = 2qt + 1$ ($t = 0, 1, 2, 3, \ldots $)  contains infinitely many
prime numbers. We are interested in primes $n$ of the indicated form with
possible small parameters $t = 1, 2, \ldots $. If the
generalized Riemann hypothesis is true, then the smallest prime number in the indicated
sequence does not exceed $c(\varepsilon) q^{2+ \varepsilon}$ for any $\varepsilon > 0$ 
($c(\varepsilon $ is a constant,
depending on $\varepsilon $). Calculation experience shows that
primes in the specified sequence occur quite
often and close to its beginning. Note also that, according to the theory known
numbers to Cramer's hypothesis $p_{n + 1} - p_n = 0 (\ln^2p_n)$ (here $p_n$
denotes the $n$th prime number in order). Roughly the same follows from the generalized Riemann hypothesis.

Suppose we need to construct a prime $p$ such that $p-1=r$ and $r$ is divisible by the product of $n+1$ pairwise coprime numbers $d$ and $t = \prod_ {i = 1}^nt_i,$ i.e., $r = dtr'$.  This can be effectively done by the process just described, by choosing the parameter $r$  dividing to $dt$ and $r'$ dividing to $q$.  Then we obtain the prime number $p = 1 + r$ and build the finite field $\mathbb{F}_p$ of order $p.$  We can assume that $d = t - 1$ which  is coprime with any number $t_i$. The order $p-1$ of the multiplicative group $\mathbb{F}_p^{\ast}$ divides to $r$. 

Therefore $\mathbb{F}_p^{\ast}$ contains $n$ cyclic subroups $T_i$ = gp($u_i$),
 where $|u_i| = t_i, i = 1, \ldots , n,$ and subgroup $F$ = gp($f$) of order $d$. Let $g$ generates $\mathbb{F}_p^{\ast}$. The elements $u_i$ are efficiently computable by the formula $u_j = g^{\frac{r}{t_j}}$.  The element $f$ is computed as $f = g^\frac{r}{d}.$
 
 Of course, there is another way to find the prime number $p$ for which $p-1$ is divisible by the product $dt $, as above. We select the even numbers $r '$ in a certain interval and check the simplicity of the number $ p = 1 + drr' $ using well-known tests (for example, the Miller-Rabin test (see \cite {Hand}). The check goes on until a simple $p$ is obtained. This method is effective and often used in practical cryptography. 
 
The indicated method of constructing the subgroups $T_i, i = 1, \ldots , n$ and $F$, as above, is obviously extended to residue rings, in particular, to rings of the form $\mathbb{Z}_n, n = pq$, where $p$ and $q$ are different primes. In this case, we can construct the primes $p$ and $q$ with the desired sets of divisors for the numbers $p-1$ and $q-1$, and then use them in our construction. See \cite {RomRSA},  \cite{RomRSAPDM} or \cite{RomCryptology} for details. 

\section{Version of encryption with hidden factors}
\label{sec:3}

   The main idea used to construct a new secret sharing protocol is the encryption scheme proposed in the works of the author \cite {RomRSA} and  \cite {RomRSAPDM}. Let $p$ and $q$ are two large different prime numbers and $n = pq$.  Encryption platform - multiplicative group $\mathbb{Z}_n^{\ast}$
of a residual ring $\mathbb{Z}_n$, in which two subgroups $F$ and $H$ of coprime exponents exp($F$) $= k$ and exp($H$) $= l$ are defined. Recall that, the exponent of a group is defined as the least common multiple of the orders of all elements of the group.  The subgroup $F$ serves as the message space, and $H$ is the space of hidden multipliers. Both of these subgroups $F$ and $H$ are publicly available. The numbers $k$ and $l$ are private. Alice, the owner of the system, computes the number $l'$ such that $ll'=1 (\bmod\, k).$ It follows, that $f^{ll'}=f$ for any element $f\in F.$ The numbers $l$ and $l'$ are private.  

Suppose Bob wants to send a message  to Alice, who is the organizer of the system. Alice will receive and decrypt this message. The algorithm works  as follows:
\begin{enumerate}
\item Bob encodes the message as $f\in F$, chooses $h\in H$ at random and sends $c = hf$ to Alice.  
\item Alice computes $$c^{ll'} = (h^l)^{l'}f^{ll'} = f. $$
\end{enumerate}

The secrecy of the proposed scheme is based on the intractability of calculating the order of an element in a finite field. 
See \cite {RomRSA},  \cite{RomRSAPDM} or \cite{RomCryptology}. 

Similar encryption is also possible on the  multiplicative group $\mathbb{F}_{p^r}^{\ast}$ of any finite field $\mathbb{F}_{p^r}$. Here we give a toy example for the scheme. 
\begin{example}
\label{exa:1}
Let $p = 8317,$ then $p-1 = 8316 = 2^2\cdot 3^3\cdot 7\cdot  11 = 108\cdot 77.$ Then $\mathbb{F}_{p}^{\ast}$ = gp($3$).  In the group $\mathbb{F}_p^{\ast}$, Alice takes two elements: $f = 3^{77} = 3113$ and $h = 3^{108} = 4610$, and defines two subgroups: $F =$ gp($f$) and $H=$ gp($h$)  of coprime exponents $108 = 2^2\cdot 3^3$ and $77 = 7 \cdot 11$ respectively. She also computes $77^{-1} = 101 (\bmod\, 108)$, more exactly: $77\cdot 101 = 7777 = 1 + 108\cdot 72.$  Then for any $m \in F$ one has $m^{108} = 1$ and $m^{7777} = m.$ For any $u \in H$ one has $u^{77} = 1.$

Suppose Bob wants to send the message $m = f^i$ ($i\in \mathbb{Z}$) to Alice. He chooses  $u = h^j$ ($j \in \mathbb{Z}$) and sends $c = um.$ Alice computes as follows: $$c^{7777} = (u^{77})^{101}m(m^{108})^{72} = m.$$

\end{example}

\begin{remark}
Of course, one can use as a platform the multiplicative group $K^{\ast}$ of any commutative associative ring $K$ with unity, provided that large subgroups of coprime exponents can be chosen in $K$, and the problem of calculating the order of an element is intractable. 
One of the advantages of this system over the original RSA version is its semantic secrecy. See \cite{RomRSA}, \cite{RomRSAPDM}  or \cite{RomCryptology} for details.
\end{remark} 
 
\section{Multiparty secret sharing protocol description}
\label{sec:4}
Let ${\mathcal A}$ be the system  which is  organized and managed by Alice.  Let 
 $ \{B_1, \ldots, B_n\} $ be the set of users in ${\mathcal A}.$
 
 \bigskip
 {\bf Version 1} 

\medskip
When setting up the system $ {\mathcal A}$, Alice takes a set of pairwise coprime positive integers $t_1, \ldots , t_n$. Let $t = \prod_{i=1}^nt_i$. Alice also defines $d = t -1$ or $d= t +1.$  Then Alice chooses a large prime finite field $\mathbb{F}_p,\, p - 1 = r,$ where $r = dtr',$ for some $r'\in \mathbb{N}$, while simultaneously defining the set of subgroups 
$T_i =$ gp($u_i$) ($i = 1,  \ldots , n$) and $F =$ gp($f$) of the multiplicative group $\mathbb{F}_p^{\ast}$ of the corresponding  orders $t_1, \ldots , t_n$ and $d$, respectively. The corresponding algorithm has been given in Section \ref{sec:2}. The prime $p$ and consequently field $\mathbb{F}_p$ are public. The subgroups $T_1, \ldots , T_n, F$ and their corresponding orders $t_1, \ldots , t_n, d$ are private. Let $H=\prod_{i=1}^nT_i$. The subgroup $F$ serves as the space of possible secrets, and $H$ is the space of hidden multipliers.

Then Alice distributes the numbers (shares) $t_1, \ldots, t_n$ among the users $B_1, \ldots , B_n$ of  $ {\mathcal A}$  corresponding to their indexes.  Each user $B_i$ receives the share $t_i.$ This distribution is carried out either over a secure communication channel, or is transmitted in encrypted form over an open channel. These shares are for future reuse.  

 Let $ m \in F $ be a message that Alice wants to send to some (qualified) set of users of the system $B(m) = B_ {i_1} \cup \ldots \cup B_{i_w}, 1 \leq i_1 < \ldots < i_{w} \leq n.$ Alice acts as follows:
\begin{enumerate}
\item Alice randomly selects nontrivial elements $v_{i_j} \in T_{i_j},\, j = 1, \ldots , w.$ Then she computes  $c = \prod_{j=1}^wv_{i_j}\cdot m^{t/\prod_{j=1}^wt_{i_j}}$.    She sends $c$ to the coalition $B(m).$
\item Members of the coalition $ B(m) $ sequentially raise the obtained element $c$ to the power  $t_j$ for  $j = 1, \ldots, w.$ In the case $d=t-1$, they get element
$$m^t = m^{d+1} = m.$$ 
If $d = t+1$, they get
$$m^t = m^{d-1} = m^{-1}$$
\noindent and compute $m.$ 
\end{enumerate}

Obviously, Alice can send the message $m$ to any possible coalition of users in this way. Any unqualified coalition will not be able to reveal the message $m$ in some natural way. If this coalition does not contain the user $B_{i_j}$ it cannot remove the factor $v_{i_j}$ and obtain also  $m^{t}.$

This scheme is not monotonous. Moreover, interference with the disclosure of a secret by any member outside the qualified coalition results in an incorrect secret. 

Let us illustrate this scheme with an example. We use the notation introduced above. 
\begin{example}
\label{exa:2}
Let  $p = 6091, p - 1 = 6090 = 2\cdot 3 \cdot 5\cdot 29 \cdot 7, t_1=2, t_2 =3, t_3 = 5, d = 29, t = 30.$ Then $\mathbb{F}_p = $ gp($2$), $u-1 = 2^{3045}= 3045, u_2 = 2^{2030}=4247, u_3 = 2^{1218} = 5842, f = 2^{210}=2901.$ The qualified coalition is $B_1\cup B_2 \cup B_3$.

Alice takes $v_1 = u_1 = 3045, v_2 = u_2^2 = 1558, v_3 = u_3^2 = 1091.$  She chooses the message $m = f^3 = 2^{630}=5948$ and sends $c = v_1v_2v_3m = 2^{4081}= 808.$ The coalition recovers the message 
as $c^{t} = m.$ 
\end{example} 

\bigskip
{\bf Version 2.} 

\medskip 
This version can be used to decrease  $d$ (while decreasing $p$). Let $d$ be a sufficiently large positive integer. Alice choose the number $t_1$ which is coprime to $d$ and computes  $t_1'$ such that $t_1t_1' = 1 (\bmod\, d).$ Note that $t_i'$ is also coprime to $d.$ Then Alice takes a random $t_2$ coprime to $d$ and computes $t_2'$ such that $ t_2t_2' = 1 (\bmod\, d)$ with the constraint that $t_2t_2'$ is coprime to $t_1t_1'$. Alice continues this process and gets a set of elements $\tilde{t}_i = t_it_i '$ for $i = 1, \ldots, n,$ all of whose members are pairwise coprime and coprime to $d$.

Now Alice takes as above the prime finite field $\mathbb{F}_p,\, p - 1 = r, r = dtr', t = \prod_{i=1}^nt_i.$  while simultaneously defining the set of subgroups 
$T_i =$ gp($u_i$) ($i = 1,  \ldots , n$) and $F =$ gp($f$) of the multiplicative group $\mathbb{F}_p^{\ast}$ of  orders $t_1, \ldots , t_n$ and $d$, respectively. The corresponding algorithm has been given in Section \ref{sec:2}. The prime $p$ and consequently field $\mathbb{F}_p$ are public. The subgroups $T_1, \ldots , T_n, F$ and their corresponding orders $t_1, \ldots , t_n, d$ are private. Let $H=\prod_{i=1}^nT_i$. The subgroup $F$ serves as the space of possible secrets, and $H$ is the space of hidden multipliers. 

Then Alice distributes the numbers (shares) $\tilde{t}_1, \ldots, \tilde{t}_n$ among the users of  $ {\mathcal A}$  corresponding to their indexes.  Each user $B_i$ receives the share $\tilde{t}_i.$ This distribution is carried out either over a secure communication channel, or is transmitted in encrypted form over an open channel. These shares are for future reuse.  

 Let $ m \in H $ be a message that Alice wants to send to some (qualified) set of users of the system $B(m) = B_ {i_1} \cup \ldots \cup B_{i_w}, 1 \leq i_1 < \ldots < i_{w} \leq n.$ Alice acts as follows:

\begin{enumerate}
\item Alice randomly selects nontrivial elements $v_{i_j} \in T_{i_j},\, j = 1, \ldots , w.$ Then she computes  $c = \prod_{j=1}^wv_{i_j}\cdot m$.    She sends $c$ to the coalition $B(m).$
\item Members of the coalition $ B (m) $ sequentially raise the  element $c$  and elements successively received from it to the power $\tilde{t}_j$ for  $j = 1, \ldots, w.$ 
\end{enumerate}

The first step gives the element 
$$c_{i_1} = v_{i_2}^{\tilde{t}_{i_1}} \cdots v_{i_w}^{\tilde{t}_{i_1}}\cdot m.$$
This means that the first factor has been removed from the record. The rest of the factors before $m$ retained their orders, since these orders are coprime to $\tilde{t}_{i_1}$.
Continuing the process, they sequentially remove all factors except $m$, which remains unchanged for all exponentiations.
As a result, they get the element $m$.
 
Unlike version 1, this version is monotonous. Any coalition containing a qualified coalition also reveals the secret. This is due to the fact that each raising to the power, which is a share of the secret, does not change the multiplier $m$. Exactly one coalition reveals the correct secret. In this case, the coalition is not obliged to know whether it is qualified in this case.

We give a toy example for the second version of the scheme. 
\begin{example}
\label{exa:3}
Let us choose as a platform the multiplicative group $\mathbb{Z}_n^{\ast}$ of the modular ring $\mathbb{Z}_n, n = pq, $ where $p = 8317$ is a prime number from Example 
\ref{exa:1}, and $q = 31.$ Then $n = 257827.$

 First, Alice solves the system of congruences: 
 $$ \begin{cases} f_1 \equiv 3113 (\bmod \, p), \\ f_1 \equiv 1 (\bmod\, q).\end{cases} $$ 
 By the Chinese Remainder Theorem she gets
  $f_1 \equiv 77966 (\bmod \, n).$ The exponent of $f_1\in \mathbb{Z}_n^{\ast}$  is equal to $108$ (the exponent of $f \in \mathbb{F}_p^{\ast}$; see Example \ref{exa:1}). The subgroup $F =$ gp($f_1$)$\in \mathbb{Z}_n^{\ast}$  is the message space.
  
In the similar way Alice solves the system: 
$$\begin{cases} h_1 \equiv 4610 (\bmod\, p), \\ h_1\equiv 1(\bmod\, q). \end{cases}$$ She finds a solution $h_1 = 71146 (\bmod\, n)$. Therefore  $h_1\in \mathbb{Z}_n^{\ast}$ has the exponent   $77$ (the exponent of $h \in \mathbb{F}_p^{\ast}$; see Example \ref{exa:1}). Then $77\cdot 101 = 7777 \equiv 1 (\bmod\, 108).$ The first share of the secret is $t_1=7777.$ For each $m \in F$,  $m^{t_1} = m$ is fulfilled. For each $u_1 \in $gp($h_1$) we get $u_1^{t_1} = 1.$ 

 Element $16 \in \mathbb{F}_{q}^{\ast}$ has exponent  $5$.  Congruences 
 $$\begin{cases} h_2 \equiv 1(\bmod\, p), \\ h_2 \equiv 16(\bmod\, q).\end{cases}$$
 \noindent  have a solution  $h_2 = 99805$. Then $h_2 \in \mathbb{Z}_n^{\ast}$ has the exponent $5$.  Since $325 = 5\cdot 65 \equiv 1 (\bmod\,108)$ the number $t_2=325$ will be the second share of the secret. For each $m \in F$ one has $m^{t_2} = m.$   For each $u_2 \in $gp($h_2$) we get $u_2^{t_2} = 1.$ 
 
Let the qualified coalition consists of two users $B_1$ and $B_2$, that have two shares of the secret $t_1$ and $t_2$ respectively.  

Suppose Alice wants to send the message $m = f^i$ ($i\in \mathbb{Z}$) to the qualified coalition $B_1\cup B_2$. She chooses  $u_1 = h_1^{s_1}$ and $u_2 = h_2^{s_2}\, (1\leq s_1, s_2 \leq 107)$. Then Alice   sends $c = u_1u_2m$ to $B_1\cup B_2$.   The coalition computes  as follows: $$c^{t_1} =u_1^{t_1}\cdot u_2^{t_1}\cdot m^{t_1} = u_2^{t_1}\cdot m;$$
$$ (u_2^{t_1}\cdot m)^{t_2} = (u_2^{t_2})^{t_1}\cdot m^{t_2} = m.$$

\end{example}

\section{ ($k, n$) - thrested scheme }
\label{sec:5}
This proposition bases on the version 2 of the multiparty sharing protocol described in Section \ref{sec:4}. 
Let us prove a preliminary statement.
\begin{lemma}
Let $n\in \mathbb{N}.$  For any $k\in \mathbb{N}_n$, there exists a set $T(k) = \{t_0, t_1, \ldots , t_{l_k}\}$, which can be represented as a union of $n$ subsets $T_j(k)$ ($j = 1, \ldots , n$) such that the union of any $k$ subsets coincides with $T(k)$, and the union of a smaller number is strictly less than $T(k)$.
\end{lemma}
\begin{proof}
Induction by $k$. The statement is true for $k = 1$, when one can define $T_j= t_0$ for any $j$. Assume that the statement is true for $k-1$ and $T(k-1) =  \{t_0, \ldots , t_{l_{k-1}}\}.$

We enumerate all subsets of the set $\mathbb{N}_n$, consisting of the $k-1$ element: 
$V_1, \ldots, V_{\binom{n}{k-1}}$. Then we take elements $t_{l_{k-1}+1}, \ldots , 
t_{l_{k-1}+ \binom{n}{k}}$ and include each $t_{l_{k-1}+i}$ into all $V_j$ except $V_i$.  Therefore we can  set 
$T(k) = \{t_0, t_1, \ldots , t_{l_k}\}$, where $l_k =  1 + \binom{n}{1} + \ldots + \binom{n}{k}$, satisfying the required condition.
 
\end{proof}
Note that for the indicated construction $ T_1 \subset T (2) \subset \ldots T(n) $ and $T (n) = 2^n-1.$

Let ${\mathcal A}$ be the system  which is  organized and managed by Alice.  Let 
 $ \{B_1, \ldots, B_n\} $ be the set of users in ${\mathcal A}.$ 
 
 \medskip
For  $s=2^n-1$, Alice takes a set of pairwise coprime positive integers $T = \{t_1, \ldots , t_s\} \cup \{d\}$.  Let $t = \prod_{i=1}^st_i.$  Then Alice chooses a large prime finite field $\mathbb{F}_p,\, p - 1 = r,$ where $r = dtr',$ for some $r'\in \mathbb{N}$, while simultaneously defining the set of subgroups 
$W_i =$ gp($w_i$) ($i = 1,  \ldots , s$) and $F =$ gp($f$) of the multiplicative group $\mathbb{F}_p^{\ast}$ of the corresponding  orders $t_1, \ldots , t_s$ and $d$, respectively. 

By Lemma 1, Alice defines a representation of the set $T$ in the form of a union of subsets $T(n)_i$
for $i = 1, \ldots , n.$ Each $T(n)_i$ is the union of ascending sets $T(1)_i \subset T(2)_i \ldots \subset T(n)_i.$ For each $j \in \mathbb{N}_n$, Alice computes $\bar{t}_i = \prod_{i\in T(n)_j}t_i$. Then she computes the shares $\tilde{t}_j = \bar{t}_j\bar{t}_j^{-}$, where $\bar{t}_j^{-} = 
\bar{t}_j^{-1} (\bmod\, d).$

Then Alice distributes the shares $\tilde{t}_j$ among the participants according to the indices.  Each participant $B_j$ receives the share $\tilde{t}_j$.  

We suppose that Alice wants to develop an ($k, n$) secret sharing scheme. In each subgroup $W_i$, where $i \leq l_m$, Alice chooses a nontrivial element $g_i$. Then she computes $$c= \prod_{i \in T(m)}g_i\cdot k$$
\noindent 
and sends this element to all participants. 

Let $C$ be a coalition, consisting of $z \geq k$ members. Coalition members consistently raise $c$ to exponents equal to their shares. Since the product of all their shares is divisible by any value $t_i $ for $i \leq l_k $, the result is the shared secret $m$. This does not happen if the coalition has fewer than $k$ members. Hence, it is a ($k, n$) secret sharing scheme. 

\bigskip
{\bf Properties and security}

\medskip
The semantic secrecy of the above schemes is based on the difficult solvability of the problem of calculating the exponent of an element in the platform under consideration (finite field or commutative associative ring with unity, in particular,  residue ring). Indeed, let there be two secrets $m_1$ and $m_2$, one of which is transmitted in the form $c = tm$, where $t$ and $m $ have coprime orders. If the attacker can calculate the exponents of the elements,  he will calculate $e_i = $exp($m_i$) and $e_i' = $ exp($c^{- 1}m_i$) for $i = 1, 2.$ Then he compares the sets of prime divisors for two pairs $ e_1, e_1 '$ and $ e_2, e_2' $. Only in the pair corresponding to the transmitted secret,  such a set for $ e_i '$ does not contain prime divisors of $e_i$.

If the problem of calculating the exponent of a protocol platform element is intractable, then this scheme is semantically secret. 

In the case of the field $\mathbb {F}_p$, to calculate the orders of the elements of the group $ \mathbb{F}_p^{\ast} $, it is sufficient to know the primary decomposition of the number $p-1 $ (see \cite{Hand}). The ability to solve the problem of calculating the order of an element of the group $\mathbb{Z}_n^{\ast}, n = pq$ ($p, q $ are different primes) gives an algorithm for calculating transmitted messages, that is, it solves the RSA problem (see \cite{RomRSA}). 

There is very little public data in the proposed schemes. In the case of the field 
$\mathbb{F}_p $, only its order $p$ is known, but the primary decomposition of the number $p-1$ is unknown. In the case of the residue ring $ \mathbb{Z}_n$, only the module $n$ is known.

 {\bf Conclusions.}

In this paper, we have studied two versions of a new secret sharing scheme based on hidden multipliers in finite fields or commutative associative rings, in particular, residue rings.  We also propose a ($k, n$) - thrested scheme, where $ k $ is chosen by the dealer for each session without any additional allocations. The main feature of the proposed schemes is that the shares of secrets are distributed once and can be used multiple times. Our schemas are secure against passive attacks and semantically secure under assumption that the problem of calculation the exponent of element of a protocol platform is intractible.  It is also easy to see that the version 2 of the sharing scheme is monotonic, but the version 1 is not. The size of each share of the secret is equal to the size of the secret, since they are platform elements.

\end{document}